\documentclass[fleqn,10pt,twocolumn]{SICE20}
\usepackage{graphicx}    
\usepackage{mathtools}      
\usepackage[square,sort,comma,numbers]{natbib}
\usepackage{fixmath}
\usepackage{amssymb}
\usepackage{amsfonts}

\usepackage{amsmath}
\usepackage{latexsym}

\usepackage[a4paper, tmargin=26mm, headsep=7mm, headheight=12pt, total={170mm, 246mm}]{geometry}

\usepackage{multirow}

\usepackage{ntheorem}

\DeclareMathOperator*{\minimize}{minimize}
\DeclareMathOperator*{\subjectto}{subject\ to}

\newcommand{\norm}[1]{\lVert #1 \rVert}

\theoremstyle{plain}

\newtheorem{thm}{Theorem}
\newtheorem{lem}[thm]{Lemma}

\newtheorem{defn}[thm]{Definition}
\newtheorem{assum}[thm]{Assumption}
\newtheorem{cor}[thm]{Corollary}

\newtheorem{prob}[thm]{Problem}
\newtheorem{prop}[thm]{Proposition}

\title{Finite-time Control of Discrete-time Positive Linear Systems \\ via Convex Optimization  }

\author{Chengyan Zhao${}^{1\dagger}$, Masaki Ogura${}^{2}$, and Kenji Sugimoto${}^{1}$}
\speaker{Chengyan Zhao}

\affils{${}^{1}$Nara Institute of Science and Technology, 
	Ikoma, Nara 630-0192, Japan\\
(Tel: +81-07-4372-5356; E-mail: \{zhao.chengyan.za5, kenji\}@is.naist.jp)\\
${}^{2}$Osaka University, 
Suita, Osaka 565-0871, Japan\\
(Tel: +81-06-6879-4358; E-mail: m-ogura@ist.osaka-u.ac.jp)\\
}
\abstract{%
In this paper, we study a class of finite-time control problems for discrete-time positive linear systems with time-varying state parameters. Although several interesting control problems appearing in population biology, economics, and network epidemiology can be described as the class of finite-time control problems, an efficient solution to the control problem has not been yet found in the literature. In this paper, we propose an optimization framework for solving the class of finite-time control problems via convex optimization. We illustrate the effectiveness of the proposed method by a numerical simulation in the context of dynamical product development processes. }


\begin{document}

\maketitle


\section{Introduction}

The concept of finite-time stability~\cite{Dorato1961}, which is concerned with the stability property of dynamical systems over a finite time window, is of practical importance due to its effectiveness in solving realistic control problems appearing in several fields including robotics~\cite{Yu2005}, spacecraft control~\cite{Du2011}, and multi-agent systems~\cite{Srinivasan2018}. We find in the literature several advances in the field; for example, Bhat and Bernstein~\cite{ref3} proposed a finite-time stability criteria for continuous-time autonomous systems. Amato et al.~\cite{ref2} proposed a sufficient condition for finite-time stability and control for time-invariant linear systems through the Lyapunov function approach. Hong~\cite{Hong2002} considered the finite-time control and stabilizability for a class of controllable systems. The authors in~\cite{Zhang2012,Haddad2008} studied a finite-time synthesis problem for the nonlinear systems.

Recently, finite-time control problems have been actively investigated in the context of positive systems~\cite{Farina2000}, which are dynamical systems whose state variables are confined to be within the positive orthant and naturally arise in various application areas including  pharmacology~\cite{Hernandez-Vargas2013}, epidemiology~\cite{Nowzari2015a,Ogura2015c}, and communication networks~\cite{Shorten2006}. For example, the authors in~\cite{Chen2012} derived a necessary and sufficient condition for the finite-time stability of switched positive linear systems by using the co-positive Lyapunov approach. Colaneri et al.~\cite{Colaneri2014} established the convexity of the norm of the state variable of a class of positive time-varying linear systems with respect to the diagonal entries of the state matrix. However, the practical applicability of this convexity result is not necessarily enough to cover some applications of positive linear systems because the convexity property is limited to the  diagonals of the state matrix of the system, as shall be discussed later in this paper. Furthermore, there is a lack of frameworks for considering the cost associated with control input such as the one for chemical~\cite{Ogura2017} and medical~\cite{Kohler2018} interventions.

Extending the framework in~\cite{OguraGP} for linear time-invariant positive systems, in this paper we propose an optimization framework to solve a class of finite-time control problems for discrete-time \emph{time-varying} positive linear systems. We formulate the finite-time control problem as an optimization problem, in which the parameter cost as well as the performance evaluation function are described by posynomial functions~\cite{ref4}. We then show that the finite-time control problem can be transformed into a geometric program, which can be efficiently solved via convex optimization. In the derivation of these results, we do not restrict the tunable entries of the system matrix to its diagonals; therefore, the contribution of this paper lies in showing a form of convexity of the problem with respect to \emph{any} of the entries of the state matrix.

This paper is organized as follows. After introducing necessary mathematical notations, in Section~\ref{section2}, we formulate the finite-time control problem studied in this paper. In Section~\ref{main_result}, we introduce our assumptions on the system and cost functions and, then, state our main result. Finally, in Section~\ref{simulation}, we illustrate the effectiveness of our results by solving the optimal resource allocation problem that arises in the context of managing product development processes.

The following notations are used in this paper. Let ${\mathbb{R}}$, ${\mathbb{R}}_+$, ${\mathbb{R}}_{++}$ denote the set of real, nonnegative, and positive numbers, respectively. Let $\mathbb{N}$ denote the set of positive integers. We let the entrywise logarithm operation $\log[\cdot] \colon {{\mathbb{R}}}^{n}_{++} \to {\mathbb{R}}^{\mathnormal n} $ be defined by $(\log[v])_i = \log{v_i}$ for all $i \in \{1, \dotsc, n\}$. Likewise, we define the entrywise exponentiation $\exp[\cdot] \colon {\mathbb{R}}^{\mathnormal n} \to {\mathbb{R}}^{\mathnormal n}_{++}$ in the same manner. We say that a matrix is nonnegative if all the entries of the matrix are nonnegative.

\section{Finite-time control problem}\label{section2}

In this section, we describe the finite-time control problem studied in this paper. In Section \ref{section2_1}, we describe the system studied in this paper and state the necessary assumptions. In Section \ref{section2_2}, we formulate the finite-time control problem as an optimization problem.

\subsection{System Model}\label{section2_1}

In this paper, we consider the following parametrized time-varying linear system defined on a finite time interval: 
\begin{equation*}
\Sigma_\theta :  x(k+1) = (A(k) + K(k;\theta(k))) x(k),~k=0, \dotsc, T,
\end{equation*}
where $x(k) \in {\mathbb{R}}^{n}$ is the state variable, $A(k) \in {\mathcal A} \subset {\mathbb{R}}^{n \times n}$ ($k=0, \dotsc, T$) is a time-varying state matrix, and $$K(k;\theta(k)) \in {\mathcal K} \subset {\mathbb{R}}^{n \times n},\  k=0, \dotsc, T,$$ is the control matrix parametrized by the vector~$\theta(k)$  belonging to a set~$\Theta \subset {\mathbb{R}}^{{\mathnormal n_\theta}}$. We assume that the set~${\mathcal K}$ is bounded. Our objective in this paper is to present an optimization framework for tuning the parameter~$\theta(k)$ in such a way that the finite-time stability of the system is guaranteed, under the positivity assumption on the system. 
The positivity of discrete-time time-varying linear systems is formally defined as follows.

\begin{defn}\cite{Farina2000}\label{den:positive}
	We say that the time-varying linear system
	\begin{equation*}
	\Sigma: x(k+1) = M(k) x(k)
	\end{equation*}
	is (\emph{internally}) \emph{positive} if for any initial condition~$x(0)$ with nonnegative entries, the corresponding state trajectory $x(k)$ is nonnegative for all $k\geq 0$. 
\end{defn}

For positive time-varying linear systems, we define the notion of finite-time stability \cite{ref1} as follows.

\begin{defn} \label{defn_FTS}
	Let $T$ be a positive integer. 
	Suppose that a positive number $\epsilon$ as well as positive vectors~$v$ and $\ell(k)$ ($k \in \{1, \dotsc, T\}$) are given. We say that 
	$\Sigma$ is \emph{finite-time stable} if the trajectory of the system satisfies 
	\begin{equation*}
	x^\top(k)\ell(k) < \epsilon,~k=1, \dotsc, T,
	\end{equation*}	
	for all initial states $x(0)$ satisfying~$x^\top(0)v \leq 1$. 
\end{defn}

In this paper, we place the following assumption on the parameterized system $\Sigma_\theta$ for ensuring its positivity. 

\begin{assum}\label{assum:positivity}
	The matrix $A+K$ is nonnegative for all $A \in \mathcal{A}$ and $K \in \mathcal{K}$.
	%
\end{assum} 


We then introduce cost and performance functions as follows.
For each control action $K(k;\theta(k))$, the control parameter $\theta(k)$ at time $k$ comes with an associated cost. In this paper, we suppose that a cost function for tuning the parameter~$\theta(k)$ is given by the following functional:
	\begin{equation*}
L \colon \mathbb R^{\mathnormal {n_\theta}} \to {\mathbb{R}} \colon \mathnormal{\theta(k)} \mapsto L(\theta(k)).
\end{equation*}

Let $x(\cdot;\theta(k))$ denote the solution of the system~$\Sigma_\theta$. In order to measure the stability property of the system, we use the functional  
\begin{equation*}
J(\theta(k)) = \norm{x(\cdot;\theta(k))}_p  
\end{equation*}
where $p>0$ is a constant and $\norm{\cdot }_p$ denotes the  $\ell^p$-norm of a sequence of real vectors.

\subsection{Problem Formulation}\label{section2_2}

In this section, we present two types of optimization problems for the finite-time control of the parametrized system~$\Sigma_\theta$. We first present the budget-constrained optimization problem to minimize $J$ while satisfying the constraint on the cost function $L$ as well as the finite-time stability. Formally, the budget-constrained optimization problem is stated as follows:

\begin{prob}\label{opt_problem}
	Let a constant $\bar L$ be given. 
	Find a sequence of variables $\theta=\{\theta(k)\}_{k=0}^T$ such that $$L(\theta)\leq \bar L$$ and the system $\Sigma_\theta$ is finite-time stable in the sense of Definition~\ref{defn_FTS}, while minimizing the cost function~$J(\theta)$. 
\end{prob}

Mathematically, the budget-constrained finite-time control problem can be stated as 
\begin{subequations}\label{eq:opt_prob:}
	\begin{align}
	&\minimize_{\theta \in \Theta^{T+1}} ~~~ J(\theta) \label{eq:opt_prob:A}\\
	&	\subjectto  ~~ x^\top(k;\theta(k))\ell(k) < \epsilon, \label{eq:opt_prob:B} \\
	& \quad	\quad \quad \quad \quad ~L(\theta) \leq \bar{L}. \label{eq:opt_prob:C}
	\end{align}
\end{subequations}

Likewise, by exchanging the roles of the objective function and constraints in the budget-constrained optimization problem, we obtain the performance-constrained optimization problem
\begin{prob}\label{opt_problem:performance_constrained}
	Let a constant $\bar J$ be given. 
	Find a sequence of variables $\theta=\{\theta(k)\}_{k=0}^T$ such that $$J(\theta)\leq \bar J$$ and the system $\Sigma_\theta$ is finite-time stable in the sense of Definition~\ref{defn_FTS}, while minimizing the cost function~$L(\theta)$. 
\end{prob}

As in \eqref{eq:opt_prob:}, we can mathematically formulate  the performance-constrained finite-time control problem as the following:
\begin{equation*}
	\begin{aligned}
	&\minimize_{\theta \in \Theta^{T+1}} ~~~ L(\theta)\\
	&	\subjectto  ~~ x^\top(k;\theta(k))\ell(k) < \epsilon,  \\
	& \quad	\quad \quad \quad \quad ~J(\theta) \leq \bar{J}. 
	\end{aligned}
\end{equation*}


\section{Main result}\label{main_result}

In this section, we present our optimization framework for solving the budget-constrained and performance-constrained finite-time control problems. Under proper assumptions, we show that the problems can be transformed into convex optimization problems.

In Problems~\ref{opt_problem} and~\ref{opt_problem:performance_constrained}, the functional $J(\theta)$ is typically a nonlinear function. Furthermore, the cost functional $L(\theta)$ is often nonlinear in applications due to the their physical characteristics such as the dosage-effect relation in the therapy control processes. For these reasons, Problems~\ref{opt_problem} and~\ref{opt_problem:performance_constrained} are not trivial to solve directly. However, in this paper, we show that a mild set of assumptions allow us to reduce the problems to geometric programming, which can be efficiently solved via convex optimization~\cite{ref4}. 

Let us first give a brief overview of geometric programming. We start from stating the following definition. 

\begin{defn}\cite{ref4}\label{defn_posynomial}	Let $v_1$, $\dotsc$, $v_n$ denote $n$ real positive variables.
	\begin{enumerate}
		\item  We say that a real function $g(v)$  is a {\it monomial} if there exist $c>0$ and $a_1, \dotsc, a_n \in {\mathbb{R}}$ such that $g(v) = c v_{\mathstrut 1}^{a_{1}} \dotsm v_{\mathstrut n}^{a_n}$.
		\item We say that a real function $f(v)$ is a {\it posynomial} if $f$ is a sum of monomials of $v$.
	\end{enumerate}
\end{defn}

The following lemma shows the log-convexity of posynomials.
\begin{lem}\cite{ref4}\label{log:convexity}
	Let $f \colon {\mathbb{R}}_{+}^{\mathnormal n} \to {\mathbb{R}}_{+} \colon \mathnormal x \mapsto f(x)$ be a posynomial function. Then, the function 
	\begin{equation*}
	F \colon {\mathbb{R}}^{\mathnormal n} \to {\mathbb{R}} \colon \mathnormal w \mapsto \log f(\textup{exp}[w])
	\end{equation*} 
	is convex. 
\end{lem}

The log-convexity of posynomials allows us to solve a class of optimization problems called geometric programs efficiently, as summarized in the following proposition~\cite{ref4}.

\begin{prop}\label{proposition}
	Let $g_1(\theta), \dotsc, g_q(\theta)$ be monomials and $f_0(\theta), \dotsc, f_p(\theta)$ be posynomials. Assume that variables $\theta \in \Theta$ satisfy Definition \ref{defn_posynomial}. We say that the following optimization problem 
	\begin{subequations}
		\begin{align*}
		&\minimize_{\theta \in \Theta} ~~~ f_0(\theta) \\
		&	\subjectto  ~~ f_i(\theta) \leq 1,~i =1, \dotsc, p, \\
		&	\quad \quad\quad\quad\quad~ g_j(\theta) = 1,~j =1, \dotsc, q, 
		\end{align*}
	\end{subequations}	
	can be transformed into a convex optimization problem through the logarithmic variable transformation
	\begin{equation*}
	\theta=\textup{exp}[z],~z \in \Gamma \subset {\mathbb{R}}^{\mathnormal m}.
	\end{equation*}
	Then, we obtain the convex optimization problem with the following form:
	\begin{subequations}
		\begin{align*}
		&\minimize_{z \in \Gamma} ~~~ \log f_0(\textup{exp}[z]) \\
		&	\subjectto  ~~ \log f_i(\textup{exp}[z]) \leq 0,~i =1, \dotsc, p, \\
		&\quad \quad\quad\quad\quad~ \log g_j(\textup{exp}[z]) = 0,~j =1, \dotsc, q.
		\end{align*}
	\end{subequations}
\end{prop}

To exploit the log-convexity of posynomials, we first place the following assumption on the structure of the parametrized time-varying linear system~$\Sigma_\theta$:

\begin{assum}\label{K_min}
	Define the matrix $K_{\min} \in {\mathbb{R}}^{\mathnormal n\times n}$ by 
	$$[K_{\min}]_{ij}=\inf\{K_{ij}\colon K \in {\mathcal K}\}.$$ Then, the matrix
	\begin{equation}\label{eq:def:tildeA}
	\tilde{A}(k)=A(k)+K_{\min}
	\end{equation}
	is nonnegative for all~$k$.
\end{assum}

We remark that, in Assumption~\ref{K_min}, the existence of the matrix~$K_{\min}$ is guaranteed by the boundedness of the set~${\mathcal K}$. Furthermore, this assumption is not very restrictive and is satisfied in the examples that we discuss in Section 4.

Using the matrix~$\tilde A$ in \eqref{eq:def:tildeA}, we rewrite the parametrized system $\Sigma_\theta$ as
$$\Sigma_\theta :  x(k+1) = (\tilde{A}(k) + \tilde{K}(k;\theta(k))) x(k),~k \in \{0, \dotsc, T\}, $$
where
\begin{equation*}
\tilde K(k; \theta(k)) = K(k; \theta(k)) - K_{\min}
\end{equation*}
is a nonnegative matrix. For this nonnegative matrix as well as the parameter space~$\Theta$, we place the following assumption. 

\begin{assum}\label{asm:K}
	The following conditions hold true:
	\begin{enumerate}
		\item 
		There exist a sequence of the posynomials $f_1(\theta)$, \dots, $f_N(\theta)$ such that 
		\begin{equation*}\label{eq:def:matlhcallV}
		\Theta = \{ \theta \in {\mathbb{R}}_{++}^{\mathnormal m} \colon  f_1(\theta)\leq 1, \dotsc, f_N(\theta)\leq 1 \}.
		\end{equation*}
		\item There exist posynomials $\kappa_{ij}(k; \theta(k))$~$(i, j \in \{1, \dotsc, n\}$ and $k \in \{0, \dotsc, T\}$) such that
		\begin{equation*}\label{eq:def:mathcalK}
		\tilde K(k; \theta(k)) = \{ [\kappa_{ij}(k; \theta(k))]_{i, j} \colon \theta \in \Theta\}
		\end{equation*}
		\item $L(\theta)$ is a posynomial. 
	\end{enumerate}
\end{assum}
\vspace{1mm}


We can now present the first main result of this paper; namely, we can show that Problem \ref{opt_problem} can be solved via convex optimization. 

\begin{thm}\label{thm_1}
	The solution of the following convex optimization problem is given by $z=\{z(k)\}_{k=0}^T$, where $z(k)$  belongs to the set~$\Gamma \subset {\mathbb{R}}^{\mathnormal m}$.
	\begin{subequations}\label{eq:opt_framework:}
		\begin{align}
		&\minimize_{z \in \Gamma^{T+1}} ~~~ \log J(\textup{exp}[z]) \label{eq:opt_framework:A}\\
		&	\subjectto  ~~ \log x^\top(k; \textup{exp}[z(k)])\ell(k) < \log \epsilon, \label{eq:opt_framework:B} \\
		& \quad	\quad \quad \quad \quad ~ \log L(\textup{exp}[z]) \leq \log \bar{L}.  \label{eq:opt_framework:C}
		\end{align}
	\end{subequations}
	Then, the solution of Problem \ref{opt_problem} is given by 
	\begin{equation}\label{exp:trans}
	\theta(k)=\textup{exp}[z(k)].
	\end{equation}
\end{thm}

\begin{proof}    
	Under the $\log$ transformation  $z(k)=\log [{\theta(k)}]$, the constraints \eqref{eq:opt_prob:A}, \eqref{eq:opt_prob:B} and \eqref{eq:opt_prob:C} are equivalent to the constraints \eqref{eq:opt_framework:A}, \eqref{eq:opt_framework:B} and \eqref{eq:opt_framework:C}, respectively. Therefore, the solution of Problem \ref{opt_problem} given by \eqref{exp:trans} becomes the solutions of optimization problem \eqref{eq:opt_framework:}.
	From Lemma \ref{log:convexity}, we can get that \eqref{eq:opt_framework:C} is convex if the cost function $L(\theta)$ follows the posynomials. Also, for the convexity of \eqref{eq:opt_framework:B}, $x^\top(k;\theta(k))\ell(k)$ is a linear function which is definitely a posynomial function. For the convexity of \eqref{eq:opt_framework:A}, we can derive that the entry of state vector is posynomials from the expansion of $x(k;\theta(k))$: 
	\begin{equation*}
	\begin{aligned}
	x(k; \theta(k))=(\tilde{A}(k-1)&+\tilde{K}(k-1; \theta(k-1))) \cdots \\ &(\tilde{A}(0)+\tilde{K}(0; \theta(0)))x(0).
	\end{aligned}
	\end{equation*}
	From the previous assumptions, we can see that if the performance measurement function follows posynomials,
	$J(\theta)$ is convex under the log transformation. From Proposition \ref{proposition}, we can see that Theorem \ref{thm_1} is a convex optimization problem.	
\end{proof}

Likewise, the performance-constrained form of finite-time control problem can also be solved through the following optimization problem:
\begin{cor}
	The solution of the following convex optimization problem is given by $z=\{z(k)\}_{k=0}^T$, where $z(k)$  belongs to the set~$\Gamma \subset {\mathbb{R}}^{\mathnormal m}$.
	\begin{equation*}
		\begin{aligned}
		&\minimize_{z \in \Gamma^{T+1}} ~~~ \log L(\textup{exp}[z]) \\
		&	\subjectto  ~~ \log x^\top(k; \textup{exp}[z(k)])\ell(k) < \log \epsilon, \\
		& \quad	\quad \quad \quad \quad ~ \log J(\textup{exp}[z]) \leq \log \bar{J}. 
		\end{aligned}
	\end{equation*}
	Then, the solution of Problem \ref{opt_problem:performance_constrained} is given by 
	\begin{equation*}
	\theta(k)=\textup{exp}[z(k)].
	\end{equation*}
\end{cor}

\section{Example: Product Development Management}\label{simulation}
In this section, we illustrate the effectiveness of our proposed framework by solving the dynamic optimal resource allocation problem for the automotive appearance design process in the car manufacturing industry.  

In this paper, we adopt the automotive appearance design example presented in \cite{Yassine2003}, which contains the following tasks: 1) carpet, 2) center console, 3) door trim panel, 4) garnish trim, 5) overhead system, 6) instrument panel, 7) luggage trim, 8) package tray, 9) seats and 10) steering wheel. Suppose there are $T$ development rounds during the process, the dynamic process of the remaining work on each task can be represented by the discrete-time positive linear system $x(k+1) = Ax(k),~k \in \{0, \dotsc, T\}$, where $x(k)$ is the remaining work vector, $A$ is the work transition matrix which is nonnegative. In this paper, we adopt the dynamic model in \cite{Yassine2016}
\begin{equation*}\label{eq:newWTM}
A_k(\phi_k, \gamma_k)=
\begin{bmatrix}
\phi_{1, k}+\Delta_{1}  &\cdots &\prod_{\ell=1}^k \gamma_{1n, \ell} \\        
\prod_{\ell=1}^k \gamma_{21, \ell} &\cdots  &\prod_{\ell=1}^k \gamma_{2n, \ell} \\   
\vdots  & \ddots &\vdots\\
\prod_{\ell=1}^k \gamma_{n1, \ell}  &\cdots  &\phi_{n, k}+\Delta_{n} 
\end{bmatrix},
\end{equation*}
where the value of the off-diagonal entries of the work transition matrix is updated with the accumulated effect in the previous investment rounds ($k-1, k-2, \dotsc , 0$). 
$\phi_k=\{\phi_{1,k},\dotsc,\phi_{n,k}\}$ represents the adjustable work efficiency of the task, while $\gamma_k=\{\gamma_{ij,k}\}, (i,j=1, \dotsc , n, i \neq j)$ are the off-diagonal entires of $A_k(\phi_k, \gamma_k)$ which represent the ratio of the extra work transferred among the tasks with progress. Furthermore, we let $\Delta_{:, k}$ to represent the abrupt change on $\phi_k$ (e.g., equipment fault, conflict on schedule or the absence of engineer). During the intermittence of the development process, the managers allocate a fixed amount of resource to prompt the development process (i.e., tuning the parameter of work transition matrix).  
We assume that the parameters can be tuned within the following intervals:
\begin{equation*}
0<\phi_{i, k}^{\textup{min}} \leq \phi_{i, k} \leq \phi_{i, k}^{\textup{max}},\quad 0<\gamma_{ij, k}^{\textup{min}} \leq \gamma_{ij, k} \leq \gamma_{ij, k}^{\textup{max}}.
\end{equation*}
Specifically, the resource can be allocated on the tasks (i.e., diagonal entries of $A_k(\phi_k, \gamma_k)$) to promote the efficiency, or on the off-diagonals to reduce the ratio of the generated work among the related tasks. 
Furthermore, suppose that the initial value of $A_k(\phi_k, \gamma_k)$ is given by $\phi_{i, k}^{\textup{max}}, \gamma_{ij, k}^{\textup{max}}$, we have to pay $f_i(\phi_{i, k})$ unit of cost for tuning the work efficiency of module $i$ from $\phi_{i, k}^{\textup{max}}$ to $\phi_{i, k}$. Likewise, we let the cost for tuning $\gamma_{ij, k}$ equal to $g_{ij}(\gamma_{ij, k})$. The total cost for the $k$th investment round is calculated by taking the sum of the cost in all the entries of $A_k(\phi_k, \gamma_k)$:
\begin{equation}\label{sim:cost}
B_k(\phi_{k}, \gamma_{k})=\sum_{i,j=1}^{n} (f_i(\phi_{i, k})+g_{ij}(\gamma_{ij, k})).
\end{equation}
Usually, a dynamic product development process contains dozens or hundreds of tasks and several investment rounds. Moreover, from the discussion in Section \ref{main_result}, the dynamic resource allocation problem for product development process is also a nonlinear optimization problem.
Thus, finding the optimal strategy is a difficult problem which can not easily be solved by the experience based method. For checking the satisfaction for Assumption \ref{K_min}, we can see that $A_k(\phi_k, \gamma_k),~k \in \{0, \dotsc, T\}$ is a sequence of nonnegative matrices with the directly tuning parameters $\phi_k, \gamma_k$ belonging to positive numbers. Then, by utilizing the knowledge in \cite{ref4}, the cost function \eqref{sim:cost} can be modeled with posynomials. Thus, the problem satisfies the assumptions and definitions in our theorem. By using Theorem \ref{thm_1}, we can transform the optimal resource allocation problem of automotive appearance design process into the finite-time control problem for positive linear system.

\begin{table*}[]
	\centering
	\caption{Work transition matrix of automotive appearance design}\label{tab:WTM_case_2}

	\begin{tabular}{c|cccccccccc|}
		\multicolumn{1}{c}{ } &\multicolumn{1}{c}{$A_{0,1}$} & \multicolumn{1}{c}{$A_{0,2}$} &\multicolumn{1}{c}{$A_{0,3}$} & \multicolumn{1}{c}{$A_{0,4}$} &
		\multicolumn{1}{c}{$A_{0,5}$} &\multicolumn{1}{c}{$A_{0,6}$} &
		\multicolumn{1}{c}{$A_{0,7}$} &\multicolumn{1}{c}{$A_{0,8}$} & \multicolumn{1}{c}{$A_{0,9}$} &
		\multicolumn{1}{c}{$A_{0,10}$}  \\
		\cline{2-11}
		\multirow{1}{*}{$A_{0,1}$} & 0.85 &0.12 & 0.02& 0.06&0.06 & & & &0.06 & \\
		
		\multirow{1}{*}{$A_{0,2}$} &0.1 & 0.53 & 0.04& & & 0.3& 0.02& &0.24 &0.02 \\
		
		\multirow{1}{*}{$A_{0,3}$} & 0.02& 0.04 &0.47 &0.08& &0.24 &0.02& &0.18&0.02 \\
		
		\multirow{1}{*}{$A_{0,4}$} &0.06 &  &0.18 &0.68& &0.14 &0.1& 0.02&0.08& \\
		
		\multirow{1}{*}{$A_{0,5}$} &0.04 &  & && 0.83& & & & &  \\
		
		\multirow{1}{*}{$A_{0,6}$} & &0.3 &0.26 &0.16& &0.28 &0.06 & & 0.02&0.2 \\
		
		\multirow{1}{*}{$A_{0,7}$} & & 0.02 &0.02 &0.1& &0.06 &0.76&0.06 &0.04& \\
		
		\multirow{1}{*}{$A_{0,8}$} & &  & &0.1 & & &0.06&0.83 &0.16&\\
		
		\multirow{1}{*}{$A_{0,9}$} &0.08 &0.24  &0.18 &0.08& &0.04 &0.04&0.16 &0.63&0.2 \\
		
		\multirow{1}{*}{$A_{0,10}$} & & 0.02 &0.02 & & &0.26 & & &0.2&0.7\\
		\cline{2-11}
	\end{tabular}	
\end{table*}

In our case study,  we select the performance-constrained problem, which aims at minimizing the total investments while satisfying the constraint on the total remaining work. For the problem initialization, we unify the initial value of remaining work with $x(0)_i=1,~(i=1, \dotsc, 10)$ (i.e., all the tasks at the beginning of development process have $100$$\%$ work remained). We set the investment rounds $T=5$ and take the sum of the remaining work after the final investment round $\sum_{i=1}^{n}x_i(T)$ as performance evaluation. Furthermore, we set the constraint value of the total remaining work with $0.001\times\sum_{i=1}^{n}x_i(0)$ (i.e., the remaining work is $0.1$$\%$ of the beginning) for judging the accomplish of process. The initial value of $A_k(\phi_k, \gamma_k)$ is given in Table \ref{tab:WTM_case_2}. Let the entries of $A_k(\phi_k, \gamma_k)$ be tuned within the  interval $[0.1, 1]$ (i.e., the component can be accelerated between $[0$$\%-90$$\%]$). Finally, we let the variance on the efficiency of each task $\Delta_{i}$ varies between $[-0.2, 0.2]$. For the finite-time stability constraint in \eqref{eq:opt_prob:B}, we set $\ell(k)=\eta(k) x(0)$, where $\eta(k)=e^{-k}, (k=1, \dotsc, 5)$, and $\epsilon=1$.
For the cost function, we adopt the following posynomial function:
\begin{equation*}\label{cost:example}
f_{ij}(\gamma_{ij})=c_{ij}\left(\frac{1}{(\gamma_{ij})^p}-\frac{1}{(\Omega_{ij})^p}\right),
\end{equation*} 
where $p>0$ is the parameter for tuning the shape of cost function, and $c_{ij}, \Omega_{ij}$, $i,j=\{1, \dotsc, 10 \}$ are positive numbers for fitting the data. For simplicity,
we unify the parameters of all cost functions with $c_{ij}=1$, $p=1$, $\Omega_{ij}=1$. In this case, for example, if $\gamma_{ij}=1$ (i.e., the corresponding entry in $A_k(\phi_k, \gamma_k)$ is not tuned), then $f_{ij}(\gamma_{ij})=0$ which means the cost is $0$. 

Fig.~\ref{fig:Remaining_work} shows that despite satisfying the object function, the dashed line does not exceed the prescribed boundedness (i.e., the designed strategy meets the constraint of finite-time stability). Through solving the convex optimization problem, we are sure to get the optimal decision variables. However, from Fig.~\ref{fig:Module} and Fig.~\ref{fig:DRs}, we can get the trends of the decision variables $\gamma$ and $\phi$, which means that the manager can foresee the trends before the process is put into effect. The information from Fig.~\ref{fig:Module} and Fig.~\ref{fig:DRs} is especially useful for the stage of product development system design, where the manager can modify the structure of the work transition matrix based on the technology of management engineering to improve the performance of the product development system via the earlier design approach.

\begin{figure}
	\begin{center}
		\includegraphics[width=8.4cm]{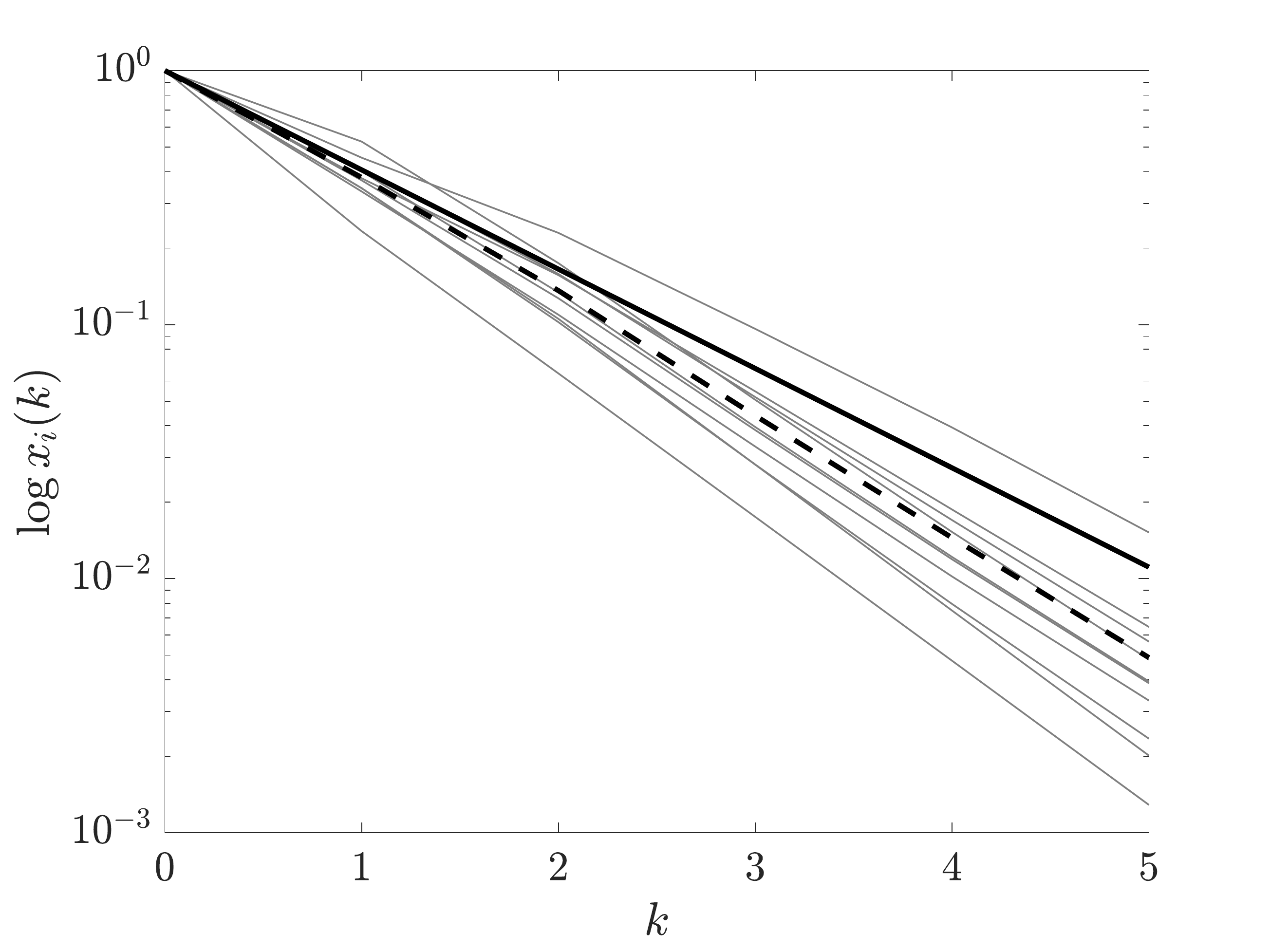}    
		\caption{Gray line: $\log x_i(k)$; Solid line: finite-time stability constraint; Dashed line: average value of $\log x_i(k)$.} 
		\label{fig:Remaining_work}
	\end{center}
\end{figure}

\begin{figure}
	\begin{center}
		\includegraphics[width=8.4cm]{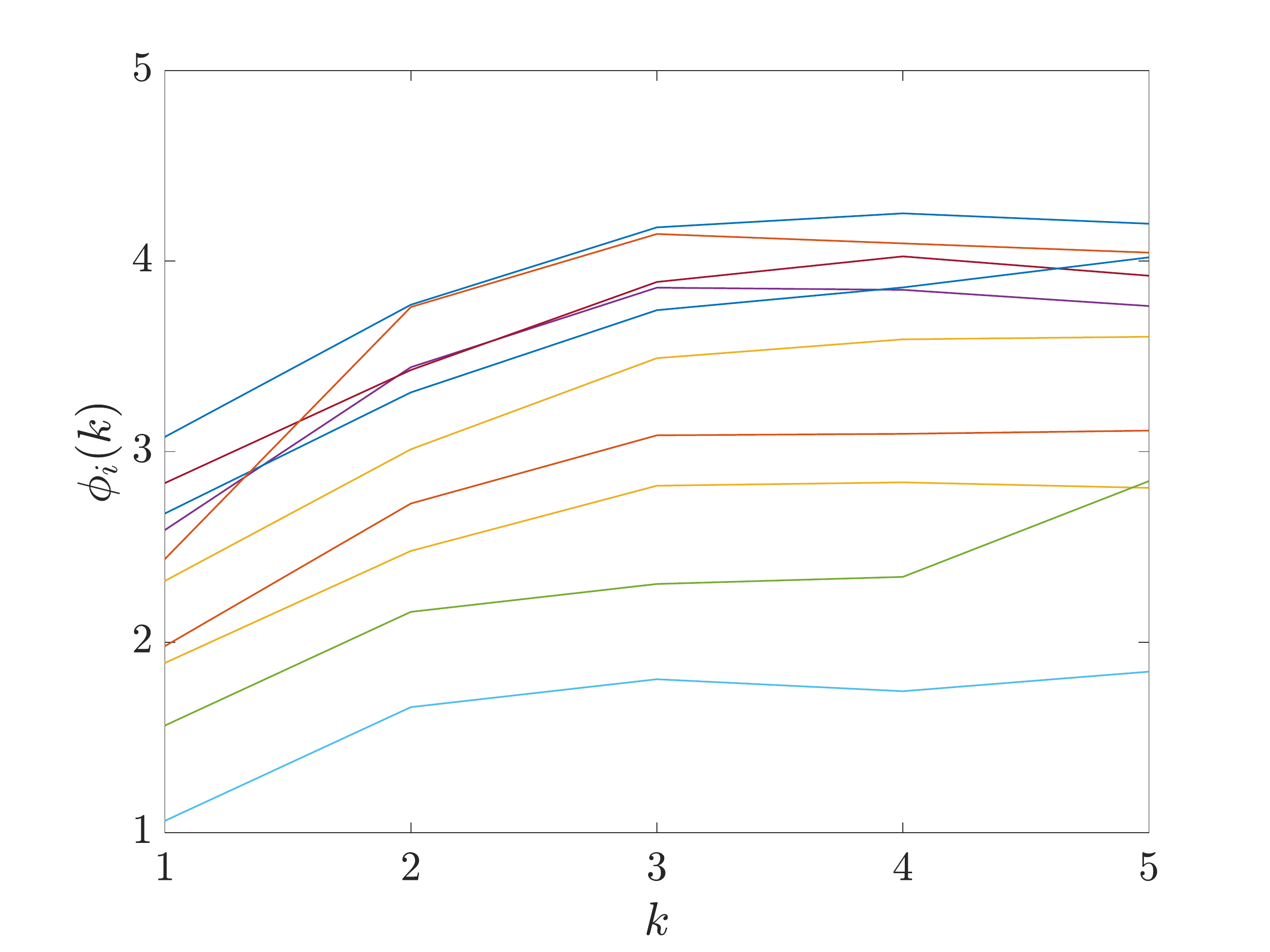}    
		\caption{The investments in $\phi_i$ versus investment round $k$.} 
		\label{fig:Module}
	\end{center}
\end{figure}

\begin{figure}
	\begin{center}
		\includegraphics[width=8.4cm]{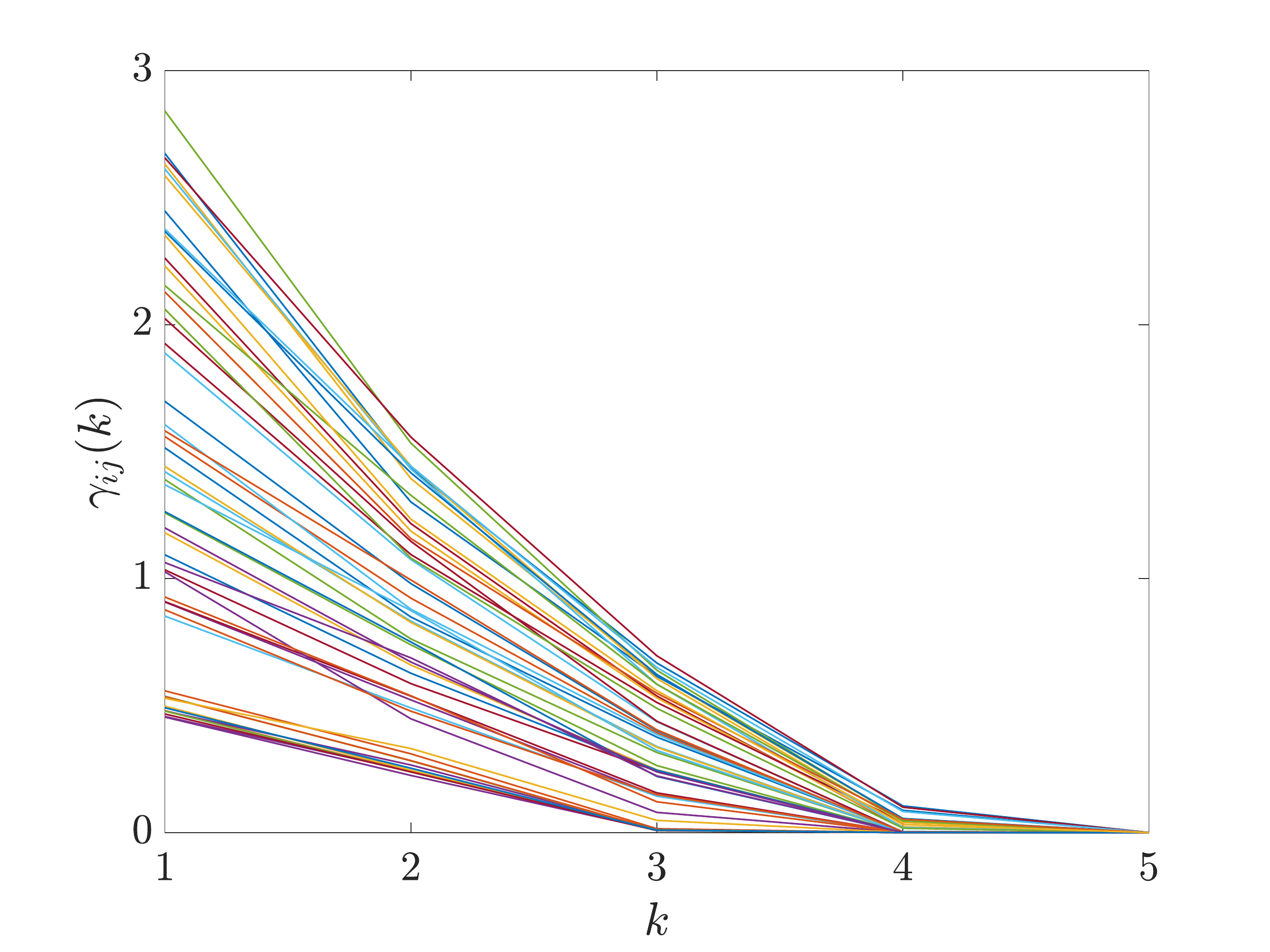}    
		\caption{The investments in $\gamma_{ij}$ versus investment round $k$.} 
		\label{fig:DRs}
	\end{center}
\end{figure}

\section{Conclusion}

In this paper, we have studied a class of finite-time control problems for the discrete-time time-varying positive linear systems constrained by the parameter tuning cost. By utilizing the convexity property of posynomial functions, we have shown that the finite-time control problem can be transformed into a convex optimization problem. Finally, we have illustrated the effectiveness of our framework by a numerical simulation on product development processes. In the future work, one of the possible extension of our work is to consider the time-delay effect; especially, if the parameters of the PD system are updated after a certain period. Then, the investment decision making problem becomes a more general situation.

%
\newcommand{\newblock}{}
\bibliographystyle{unsrt}

\end{document}